\documentclass[9pt,shortpaper,twoside,web]{ieeecolor}
\usepackage{generic}
\usepackage{cite}
\usepackage{amsmath,amssymb,amsfonts}
\usepackage{algorithmic}
\usepackage{graphicx}
\usepackage{textcomp}
\usepackage{algorithm,algorithmic}
\usepackage{amsmath}
\usepackage{amssymb}
\usepackage{graphicx}
\usepackage{comment}
\DeclareMathOperator*{\argmax}{argmax} 
\newtheorem{definition}{Definition}
\newtheorem{theorem}{Theorem}

\newtheorem{lemma}{Lemma}
\newtheorem{remark}{Remark}
\newtheorem{assumption}{Assumption}
\newtheorem{proposition}{Proposition}
\def\BibTeX{{\rm B\kern-.05em{\sc i\kern-.025em b}\kern-.08em
    T\kern-.1667em\lower.7ex\hbox{E}\kern-.125emX}}
\begin{document}
\title{A Tractable Truthful Profit Maximization Mechanism Design with Autonomous Agents}
\author{Mina Montazeri, Hamed Kebriaei, \IEEEmembership{Senior Member, IEEE,} and Babak N. Araabi
\thanks{Mina Montazeri, Hamed Kebriaei and Babak N. Araabi are with the School of ECE, College of Engineering, University of Tehran, Tehran, Iran. Emails: \{mina.montazeri@ut.ac.ir, kebriaei@ut.ac.ir, araabi@ut.ac.ir\}. }}

\maketitle

\begin{abstract}
Task allocation is a crucial process in modern systems, but it is often challenged by incomplete information about the utilities of participating agents. 
{
In this paper, we propose a new profit maximization mechanism for the task allocation problem, where the {task publisher} seeks an optimal incentive function to maximize its own profit and simultaneously ensure the truthful announcing
of the agent’s private information (type) and its participation in
the task
, while an autonomous agent aims at maximizing its own utility function by deciding on its participation level and announced type.} 
Our mechanism stands out from the classical contract theory-based truthful mechanisms as it empowers agents to make their own decisions about their level of involvement, making it more practical for many real-world task allocation scenarios.
It has been proven that by considering a linear form of incentive function consisting of two decision functions for the task publisher the mechanism's goals are met. The proposed truthful mechanism is initially modeled as a non-convex functional optimization with the double continuum of constraints, nevertheless, we demonstrate that {by deriving an equivalent form of the incentive constraints,} it can be reformulated as a tractable convex optimal control problem. Further, we propose a numerical algorithm to obtain the solution.

\end{abstract}


\section{Introduction}
\label{Sec:Intro}
Task allocation is an essential aspect of many modern systems, including supply chain management \cite{xiao2012modeling}, transportation \cite{ye2017fair}, and distributed computing \cite{caicedo2011distributed}. It involves the assignment of tasks to the agents based on their capabilities and availability, with the goal of maximizing efficiency and achieving the desired outcome. { However, in most of real-world task allocation applications, the task publisher faces several challenges like incomplete information about the agents' utilities, and the autonomy of the agents in deciding on their own participation levels. These challenges hinder optimal decision making in the task allocation problem for the task publisher \cite{dai2019task}.}

In this paper, we examine the task allocation problem in the presence of incomplete information. One approach that has been proposed in the literature for addressing this problem is the Bayesian Stackelberg game. In a Bayesian Stackelberg game, the task publisher (leader) is uncertain about the agent's type which is a parameter (or some parameters) in the agent's (follower) objective function \cite{paruchuri2008efficient}.
Thus, the leader maximizes its own expected payoff with respect to the distribution of the follower agent's type, subject to the best response of the follower. 
There are plenty of research works that have addressed the Bayesian Stackelberg game in task/resource allocation problems, including, e.g. power allocation problems \cite{duong2015stackelberg}, demand response \cite{yang2018demand}, and crowdsensing \cite{nie2020incentive}.


While the researchers on Bayesian Stackelberg games consider that the task publisher has incomplete information about the agents' utility function, they still have some unrealistic assumptions. Firstly, they assume that an agent voluntarily participates in
the game. This is not practical in applications like task allocation, where the cost incurred by doing the task may result in a negative payoff for the agent.
Secondly, in the Bayesian Stackelberg game, the optimal strategy of the task publisher is obtained in the average sense with respect to the distribution of the agent's type, and hence, {the agent receives an incentive function from the task publisher which is not necessarily designed in accordance with its actual type.} 
Economic theory provides an elegant tool called {``mechanism design"} to address such challenges.

Mechanism design offers a framework, especially for task allocation problems under asymmetric information in which, the task publisher is not aware of the agents’
private information \cite{blumrosen2007algorithmic}.
There are two kinds of mechanisms in the literature: direct mechanisms and indirect mechanisms. In the direct revelation mechanism, each agent is asked to announce its private information. While, in the latter, agents don’t announce their private information directly and agents' preferences can be observed only indirectly through their decisions. 

In a direct revelation mechanism for the task/resource allocation problem, the only action available to the agents is to announce their types. In this case, the {task publisher} allocates a participation level and the corresponding incentive reward to each agent as the functions of its announced type, in order to achieve three objectives, simultaneously: motivate agents to participate in the task, ensure truthful announcing of the private information of agents, and maximize the task publisher’s utility (or maximize the social welfare) \cite{montazeri2022distributed}. The direct mechanism may or may not induce a game among the agents. The latter is called also contract theory \cite{bolton2004contract}, while VCG is a well-known example of the former. Some direct mechanisms also consider further properties like (weak) budget balance \cite{kosenok2008individually} or multi-dimensional private information \cite{tavafoghi2016multidimensional}. 
{However, in many resource/task allocation applications, the agents may prefer to control the feature related to the participation level by themselves, which occurs in indirect mechanisms.} There are several pieces of research in the literature that design indirect mechanisms for resource/task allocation problems. Most of the proposed mechanisms satisfy criteria, such as Nash optimality, budget balance, or individual rationality \cite{kakhbod2011efficient,heydaribeni2019distributed} while some of them provide an algorithm for their mechanism to reach the equilibrium point \cite{eslami2022incentive,farhadi2018surrogate}. 
{ While the agents in these indirect mechanisms are not asked to announce their private information, this kind of mechanism is applicable for social welfare maximization, assuming the designer is not a profit maximizer.}
In this paper, we propose a new profit maximization mechanism for the task allocation problem to achieve the following goals (i) Maximize the task publisher's utility function (as in direct mechanisms, contract theory, and Stackelberg game and unlike indirect mechanisms) (ii) Let autonomous agents decide on their own participation level (as in Stackelberg game, and indirect mechanisms and unlike direct mechanisms) (iii) Guarantee the participation of agents and also truthful announcing of agent's type in the mechanism (as in direct mechanisms and unlike indirect mechanisms, and Stackelberg games).
{To the best of the authors’ knowledge, this paper introduces
the first truthful profit maximization mechanism that achieves these three goals, simultaneously.
}
 
 {
In the proposed mechanism, the task publisher seeks an optimal incentive reward as a function of the agent’s announced type and participation level, in order to maximize its own profit. After the reward function is imposed on the agents, each autonomous agent determines its optimal announced type and participation level. Therefore, the optimization problem of the task publisher is subject to some constraints which are: the agents' best response to the reward function, { non-negative profit making of the agents out of participation in the task (Individual Rationality)}, and making the truthful announcing of the type as the best {strategy} of the agents to the reward function (Incentive Compatibility).
We show that, by using a linear form of the incentive reward function, including two decision functions, the task publisher achieves these goals, simultaneously. 
In the designing process, we consider that the type of each agent is drawn from a specific continuous distribution that is known to the task publisher. However, by considering continuous distribution for the agents' type, we face a non-convex double continuum of incentive constraints which brings another theoretical challenge to our problem.
We show that the optimal profit maximization mechanism that satisfies all the aforementioned properties is obtained by solving a constrained nonconvex functional optimization. {
Then, a relation between the decision functions of the task publisher is derived, which is used to reformulate the functional optimization as a tractable convex optimal control problem.} Finally, we propose a numerical algorithm to obtain the solution.
}
{The main contributions of the paper are as follows:
\begin{itemize}
\item We propose a new truthful mechanism for the task allocation problem that allows autonomous agents to selfishly decide on their own participation levels while ensuring the truthful announcing
of the agent’s private information (type), its participation in
the task, and maximizing the task publisher's profit.
\item We prove that by considering a linear form of the incentive reward function,  including two decision functions for the task publisher, all the mentioned properties can be satisfied. 
\item 
{ 
We prove that by introducing a relation between the decision functions of the task publisher, the main non-convex functional optimization problem can be reformulated to an equivalent tractable convex optimal control problem.}

\end{itemize}
}


\textit{Notation:}
The symbols $\mathbb{R}$, $\mathbb{R}^{+}$, and $\mathbb{R}^n$ denote real numbers, positive real numbers, and the set of n-dimensional real column vectors, respectively; $\boldsymbol{1}$ denotes the all-ones vector. 
For a given vector or matrix $X$, $X^T$ denotes its transpose. 
Given a set $ \mathbb{U} $  and a point $y$, the projection of $y$ onto $ \mathbb{U} $, denoted by $P_U(y) \in \mathbb{U}$ satisfies $	\|y-P_U(y)\| \le \|y-v\| \;\;\; \forall v \in \mathbb{U}$.
 $\big(\frac{\partial f}{\partial x}\big)_{x^*}=\frac{\partial f(x)}{\partial x}\Big|_{x=x^*}$ denotes the partial derivative of $f(x)$ with respect to $x$  at point $x = x^*$. 

\section{profit maximization mechanism}
\label{Sec:Problem_Formulation and model}
{
We consider a task allocation problem comprising two main parties: the task publisher and the autonomous agents.
}
Inspired by \cite{xiong2017economic,xiong2020contract},
the utility function of the agent is formulated as follows
\begin{align}
\label{U_MU}
U(\theta,x,R(x,\hat{\theta})) =S(x,\theta) + R(x,\hat{\theta})
\end{align}
where {$x \in \mathbb{R}^{+}$} is the level of participation of the agent in the task,
 $\theta\in\Theta$ with $\Theta=[\underline{\theta}, \bar{\theta}]$ and {$\underline{\theta}, \bar{\theta}>0$} represents the level of the agent's willingness to participate in the task which is the private information of the agent and is treated as its type and $\hat{\theta}\in\Theta$ is the announced type which is not necessarily equal to the actual type (i.e. $\theta$) as agents may have the incentive to announce their type incorrectly if it makes more profit for them.
 {Although neither the task publisher nor other agents know the agent's type, its cumulative distribution $F(\theta)$ is common knowledge.}
 Function
$R(x,\hat{\theta}):\mathbb{R}^{+} \times \Theta \to \mathbb{R}^{+} $ is the incentive reward function that each agent receives from the task publisher. 
$ S(x,\theta)=\theta \pi(x)- px$ is the {satisfaction function} of the agent from participating in the task and it includes two parts: the first term, i.e., $ \theta \pi(x): \mathbb{R}^{+} \times \Theta \to \mathbb{R}$, represents the revenue of the agent obtained from participation in the task and the second term, i.e., $px$, represents the cost incurred by doing tasks with {$p>0$} as the linear cost coefficient or marginal cost.  


The task publisher's {utility} function is defined
as follows
\begin{align}
\label{U_CP}
V(x, R(x,\hat{\theta})) =g(x)-R(x,\hat{\theta}).
\end{align}
where $g: \mathbb{R}^{+}  \to \mathbb{R}$ is the task publisher's revenue function from agents' participation in the task. 
In summary, the task publisher designs an incentive reward as the function of the agent’s announced type and the agent’s participation level and sends it to the agents. After that,  Each agent determines its optimal announced type and participation level based on the incentive reward function and communicates it to the task publisher to finalize the contract. The information flow between the task publisher and autonomous agents is depicted in Figure \ref{Fig:system_model}.
\begin{figure}[h!]
	\centering
	\includegraphics[width=0.95\linewidth]{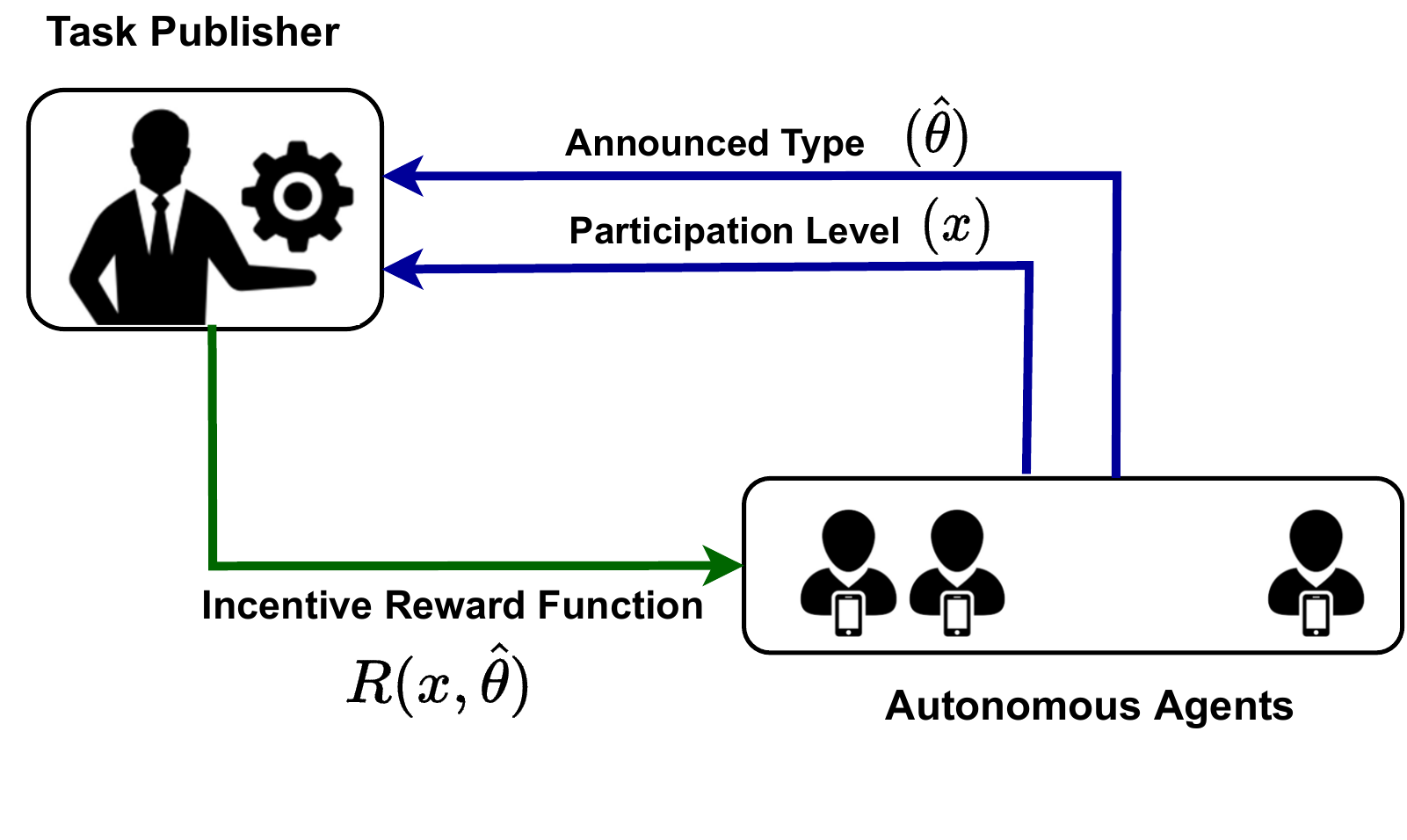}
	\caption{{The information flow between task publisher and autonomous agents.}
	}
	\label{Fig:system_model}
\end{figure}
In this study, we consider settings without externalities, where the payoffs received by each agent only depend on their level of participation and there is no game among the agents.
\begin{remark} 
The proposed model can also be extended to handle $n$ independent tasks with $\textbf{x}=(x_1,\dots,x_n )^\top \in \mathbb{R}^n$. In this case, both the agent's and task publisher's utility functions
which are $U(\boldsymbol{\theta},\textbf{x},\boldsymbol{R(x,\hat{\theta})})=\boldsymbol{\theta}^T\boldsymbol{\pi}(\textbf{x})-\textbf{p}^T\textbf{x}+\boldsymbol{1}^T\boldsymbol{R(x,\hat{\theta})}$ and $V(\textbf{x},\boldsymbol{R(x,\hat{\theta})}) =\boldsymbol{1}^T\boldsymbol{g}(\textbf{x})-\boldsymbol{1}^T\boldsymbol{R(x,\hat{\theta})}$, respectively, 
can be decomposed into $n$ independent utility functions. This is achieved by defining $\boldsymbol{\theta}=({\theta}_1,\dots,{\theta}_n )^\top \in \mathbb{R}^n$ and similarly $ \boldsymbol{\hat{\theta}},\boldsymbol{\pi}(\textbf{x}),\boldsymbol{g}(\textbf{x}),\textbf{p}$ and $\boldsymbol{R(x,\hat{\theta})}$ as vectors of $\hat \theta_i,\pi_i(x_i),g_i(x_i), p_i, R_i(x_i,\theta_i)$, respectively.

\end{remark}
As is customary in the literature \cite{xiong2020contract}, we assert the following assumption on functions $g$ and $\pi$.
\begin{assumption} 
\label{as:conc}
	 Functions $\pi(\cdot)$ and $g(\cdot)$ are non-decreasing and strictly concave. 
\end{assumption}
{In the proposed profit maximization mechanism, the task publisher adopts the following form of the reward function.
\begin{align}
\label{reward funtion}
R(x,\hat{\theta})\equiv \alpha(\hat{\theta}) x + \beta(\hat{\theta})
\end{align} 
where $\alpha(\hat{\theta})$ and $\beta(\hat{\theta})$ are both decision functions of the task publisher, which respectively represent the reward factor for the participation level paid by the task publisher and the bias reward from the task publisher to the agent. We show that this linear form of incentive reward function is rich enough to enable the task publisher to maximize its profit while {satisfying} three constraints simultaneously: (i) react optimally to the best response of the agents, (ii) motivate agents to participate in the task, 
(iii) ensure truthful announcing of the agent's type, i.e. $\hat{\theta}=\theta$. 
The formal definitions of the last two goals are presented in what follows. 
}
\begin{definition}
	\label{Def:IR}
	A mechanism is Individually Rational (IR),
	if the agent's utility is non-negative by truthful announcing of the type, i.e. $\hat{\theta}=\theta$.
 Specifically, the mechanism is IR if
   \begin{align}
U({\theta} ,x,\alpha({\theta}), \beta({\theta})) \ge 0  .
\label{Eq:IR} 
\end{align}
\end{definition}

\begin{definition}
	\label{Def:IC}
	A mechanism is Incentive Compatible (IC) 
	if the agent achieves maximum utility by truthful announcing of the type, i.e. $\hat{\theta}=\theta$.
 Specifically, the mechanism is IC if
   \begin{align}
U({\theta} ,x,\alpha({\theta}), \beta({\theta})) \ge U({\theta} ,x,\alpha(\hat{\theta}), \beta(\hat{\theta}))\;\;\; \forall  \hat{\theta} \in  [\underline{\theta}, \bar{\theta}]  .
\label{Eq:IC} 
\end{align}

\end{definition}
The task publisher's goal is to maximize its utility subject to the agent's best response, the Individual Rationality (IR), and the Incentive Compatibility (IC) constraints. 
Therefore, 
the optimal profit maximization mechanism can be obtained by solving the following
maximization problem
\begin{subequations}
	\label{opt_1_1}
\begin{align}
\label{cost}
\max_{\alpha({\hat{\theta}} ), \beta({\hat{\theta}} )}&\mathbb{E}_{\theta}[V(\chi,\alpha(\hat{\theta}), \beta(\hat{\theta}))]\\
    \label{c3}
 s.t. \; \;  & \chi(\theta, \alpha({\hat{\theta}} ),\beta({\hat{\theta}}))=\argmax_{x} U({\theta} ,x,\alpha({\hat{\theta}} ),\beta({\hat{\theta}}))
 \\ 
 \label{c1}
    &U({\theta} ,\chi,\alpha({\theta}), \beta({\theta})) \ge 0, \:\:\: \:\:\: \:\:\:\: \forall {\theta} \in  [\underline{\theta}, \bar{\theta}]
    \\
   \label{c2}
   &U({\theta} ,\chi,\alpha({\theta}), \beta({\theta})) \ge U({\theta} ,\chi,\alpha(\hat{\theta}), \beta(\hat{\theta})),
  \;\; \forall {\theta},\hat\theta \in  [\underline{\theta}, \bar{\theta}]
  .
\end{align}
\end{subequations}
{
Based on the constraints \eqref{c3} and \eqref{c2}, the task publisher takes into account the best response of the agents on the variables $x$ and $\hat \theta$, respectively and further,
according to \eqref{c2}, the reward function is designed such that the optimal announced type by the agents, is equal to their actual type $\theta$.}
\begin{remark}
Based on \eqref{opt_1_1}, the proposed mechanism maximizes the utility of the task publisher, rather than social welfare maximization in some direct mechanisms and also other indirect mechanisms \cite{eslami2022incentive,ma2020incentive}.
Furthermore, in contrast to the existing profit maximization mechanisms (contract theory) and similar to the Stackelberg game, the participation level $x$ is decided by the autonomous agent, and the task publisher takes into account the best response of the agent as expressed in constraint \eqref{c3}.  However, unlike the Stackelberg game and similar to direct mechanisms, the proposed mechanism ensures both truthful announcing of the types by the agents and their rational participation in the mechanism as expressed in constraints \eqref{c2} and \eqref{c1}, respectively.
\end{remark}

For ease of notation,  we get $\chi=\chi(\theta, \alpha({\hat{\theta}} ),\beta({\hat{\theta}}))$ when we do not want to indicate that $\chi$ is a function of $\theta$, $\alpha({\hat{\theta}} )$, and $\beta({\hat{\theta}})$.

\section{tractable reformulation of the mechanism }
\label{Sec:Contract}
 The solution of optimization \eqref{opt_1_1} gives the optimal reward function that maximizes the task publisher's utility subject to the best response of the agent,
IC and IR constraints. 
However, solving optimization \eqref{opt_1_1} is not straightforward due to the nonconvex double continuum constraint imposed by the IC constraint \cite{barucci2000incentive}. 
 To tackle this issue, in this section, we follow a multi-step approach to obtain an \textit{``equivalent tractable reformulation"} for optimization \eqref{opt_1_1}.
 At first, we reformulate the constraints \eqref{c3} and \eqref{c1}. {Next, a relation between the functions  $ \beta({\hat{\theta}} )$ and $ \alpha({\hat{\theta}} )$ is derived, which proves that is equivalent to the incentive constraints \eqref{c2}. This helps us to reformulate the initially nonconvex functional optimization into a tractable convex optimal control problem.} Finally, we propose a numerical algorithm to obtain the solution.


\begin{proposition}
\label{lambda_def}
  Constraint \eqref{c3} is equivalent to
\begin{align}
 \label{d11}
 \chi(\theta, \alpha(\hat{\theta}))=\Gamma(\frac{p-\alpha(\hat{\theta})}{\theta})
\end{align}
where function { $\Gamma(.)$} is the inverse function of $\big(\frac{\partial \pi}{\partial x}\big)_\chi$.
\end{proposition}
\begin{proof}
{In \eqref{c3}, $\chi$ is a critical point for $U({\theta} ,x,\alpha({\hat\theta}), \beta({\hat\theta}))$ over $x$ {and from Assumption \ref{as:conc}, $U$ is a strictly concave function in $x$}. Thus,
 using the first-order optimality condition we have}
\begin{align}
\label{Eq:det_x_4}
\big(\frac{\partial \pi}{\partial x}\big)_\chi= \frac{p-\alpha(\hat{\theta})}{\theta}.
\end{align}
 From Assumption \ref{as:conc} $\frac{\partial \pi(x)}{\partial x}$ is a strictly monotone function. Thus, we can define the function { $\Gamma(.)$} as the inverse function of $\big(\frac{\partial \pi}{\partial x}\big)_\chi$, which gives $\chi(\theta,\alpha(\hat{\theta}))=\Gamma(\frac{p-\alpha(\hat{\theta})}{\theta})$.
\end{proof}
\begin{remark}
\label{rem1}
	According to Assumption \ref{as:conc}, $\frac{\partial \pi(x)}{\partial x}$  is a decreasing function. Thus, function $\Gamma$ as an inverse function of $\big(\frac{\partial \pi}{\partial x}\big)_\chi$, is decreasing with respect to $\frac{p-\alpha(\hat{\theta})}{\theta}$. Hence 1) $\frac{\partial \chi}{\partial\theta}>0$, and 2) $\frac{\partial \chi}{\partial\alpha}>0$.
\end{remark}
 Proposition \ref{Lem1} shows that the IR constraint \eqref{c1} can be reduced to an equality constraint for  $\theta=\underline \theta$.  
\begin{proposition}
	\label{Lem1} 
	If the following conditions hold:
 \begin{enumerate}
     \item Constraint \eqref{c1} is satisfied for $\theta=\underline \theta$, 
     \item Constraint \eqref{c2} is satisfied for all ${\theta}\in [\underline{\theta}, \bar{\theta}]$,
 \end{enumerate}
then \eqref{c1} is also satisfied for every $\theta> \underline \theta$.
Furthermore, in any optimal solution of optimization \eqref{opt_1_1}, we have { $U(\underline{\theta},\chi,\alpha(\underline{\theta}),\beta(\underline{\theta})) = 0$}, meaning that the IR constraint is active for $\underline{\theta}$.
\end{proposition}
\begin{proof}
	By substituting $\chi$ from \eqref{d11} in $U$, and differentiating $U$ with respect to agent's actual type, i.e. $\theta$, we have
 \begin{align}
		\label{Eq:IR_pr1}
	\frac{d U(\theta,\chi,\alpha(\hat{\theta}), \beta(\hat{\theta}))}{d \theta}=  \big(\frac{\partial U}{\partial x}\big)_\chi\frac{\partial  \chi(\theta,\alpha(\hat{\theta})) }{\partial \theta}+ \pi(\chi).
		\end{align}
The first term of \eqref{Eq:IR_pr1} equals zero due to the first-order condition, i.e. $\big(\frac{\partial U}{\partial x}\big)_\chi=0$, and hence we have 
   \begin{align}
\label{Eq:IR_pr2}
\frac{d U(\theta,\chi,\alpha(\hat{\theta}), \beta(\hat{\theta}))}{d \theta}=\pi(\chi) \ge 0.
\end{align}  
Let us consider the agent with type $ \tilde{\theta} \in  [\underline{\theta}, \bar{\theta}]$ and use the following inequality
 \begin{align}
\nonumber
 &U(\tilde{\theta},\chi(\tilde{\theta}, \alpha(\tilde{\theta})),\alpha(\tilde{\theta}), \beta(\tilde{\theta})) \ge 
 U(\tilde{\theta},\chi(\tilde{\theta}, \alpha(\underline{\theta})),\alpha(\underline{\theta}), \beta(\underline{\theta}))\\
 \label{Eq:IR_pr3}
 &\ge
  U(\underline{\theta},\chi(\underline{\theta}, \alpha(\underline{\theta})),\alpha(\underline{\theta}), \beta(\underline{\theta}))
\end{align}
where the first inequality holds from IC constraint \eqref{c2} for 
$\hat \theta=\underline{\theta}$, and the second inequality follows from \eqref{Eq:IR_pr2}. 
Thus, it follows $ U(\tilde{\theta},\chi(\tilde{\theta}, \alpha(\tilde{\theta})),\alpha(\tilde{\theta}), \beta(\tilde{\theta})) \ge  U(\underline{\theta},\chi(\underline{\theta}, \alpha(\underline{\theta})),\alpha(\underline{\theta}), \beta(\underline{\theta}))$. 
Hence, if the constraint \eqref{c1} is satisfied for $\theta=\underline \theta$ and constraint \eqref{c2} is satisfied for all ${\theta}\in  [\underline{\theta}, \bar{\theta}]$, then IR constraint \eqref{c1} is satisfied for every $\theta> \underline \theta$.
{ To complete the proof, it must be shown that $IR_{\underline{\theta}}$ is binding. If $IR_{\underline{\theta}}$ is not bind, the value of $\beta(\tilde{\theta})$ can be decreased by a sufficiently small $\epsilon>0$ for all $\tilde{\theta} \in  [\underline{\theta}, \bar{\theta}]$ such that, the task publisher’s utility increases while both $IR_{\underline{\theta}}$ constraint and the IC constraints are still satisfied. The former is true due to the strict $IR_{\underline{\theta}}$  inequality \eqref{c1}, and the latter holds since subtracting $\epsilon$ from both sides \eqref{c2} does not change the IC inequality. This contradicts with optimally of solution. }
\end{proof}

Next, the optimization \eqref{opt_1_1} is reformulated by introducing a relation between decision functions $\alpha(\hat \theta)$ and $\beta(\hat \theta)$. This results in removing the non-convex IC constraints.

\begin{theorem}
	\label{Lem2}
 Optimization \eqref{opt_1_1} is equivalent to the following optimization problem.
 \begin{subequations}
	\label{opt_2_2}
\begin{align}
\label{cost1}
&\max_{\alpha(\hat{\theta}), \beta(\hat{\theta})}\mathbb{E}_{\theta}[V(\chi, \alpha(\hat{\theta}), \beta(\hat{\theta}))]\\
 \label{c1_2}
  & s.t. \,\, \,\dot{\alpha} ({\hat{\theta}}) \ge 0 
  \\
   \label{c2_2}
   &\beta({\hat{\theta}})=   \int_{\underline{\theta}}^{\hat{\theta}}[K_{\theta}(\chi(\theta,\hat{\theta}),\theta,\hat{\theta})]\Big|_{\theta=\hat{\theta}}d\hat{\theta}-K(\chi(\theta,\alpha(\hat{\theta})),\theta,\hat{\theta})\Big|_{\substack{
		{\theta={\hat{\theta}}}
}}\\
    \label{c3_2}
& \chi(\theta,\alpha(\hat{\theta}))=\Gamma(\frac{p-\alpha(\hat{\theta})}{\theta})
\end{align}
\end{subequations}
where $K(\chi(\theta,\alpha(\hat{\theta})),\theta,\alpha(\hat{\theta}))\equiv\alpha(\hat{\theta})\chi(\theta,\alpha(\hat{\theta}))+\theta\pi(\chi(\theta,\alpha(\hat{\theta})))-p\chi(\theta,\alpha(\hat{\theta}))$ and $K_{\theta}$ defines the derivation of $K$ with respect to the actual type of agent, i.e. $\theta$.
\end{theorem}
\begin{proof}
To show the equivalency of optimization problems \eqref{opt_1_1} and \eqref{opt_2_2}, it suffices to prove that for any optimal solution to problem \eqref{opt_1_1}, there exists a solution to problem \eqref{opt_2_2} with the same objective value and visa versa. We prove this theorem in two steps. First, we show that given a solution to \eqref{opt_2_2}, we can find a corresponding solution to \eqref{opt_1_1} with the same objective value.
By considering the definition of agent's utility in \eqref{U_MU}, $K$ can be rewritten as
		\begin{align}
	\label{U_K}
	K(\chi(\theta,\alpha(\hat{\theta})),\theta,\alpha(\hat{\theta}))=U(\theta,\chi(\theta,\alpha(\hat{\theta})),\alpha(\hat{\theta}), \beta(\hat{\theta}))-\beta(\hat{\theta}).
	\end{align}
	By derivation from \eqref{U_K} with respect to  $\theta$, we have
			\begin{align}
	\label{K_derivation}
	\frac{dK}{d\theta}=\big(\frac{\partial U}{\partial x}\big)_\chi\frac{\partial \chi}{\partial \theta}+\frac{\partial U}{\partial \theta}-\frac{\partial \beta(\hat{\theta}) }{\partial \theta}.
	\end{align}
	Since $\beta$ is not a function of $\theta$ and $\big(\frac{\partial U}{\partial x}\big)_\chi=0$,  
	we obtain
		\begin{align}
	\label{K_derivation2}
	\frac{dK}{d\theta}=\pi(\chi).
	\end{align}
Now considering an agent with type $\tilde{\theta}$, 
	  by replacing $\beta$ from constraint \eqref{c2_2} and $\frac{dK}{d\theta}$ from \eqref{K_derivation2}, IC constraints \eqref{c2} can be rewritten as
	\begin{align}
	\nonumber
	&\alpha(\tilde{\theta})\chi(\tilde{\theta},\alpha(\tilde{\theta}))+\int_{\underline{\theta}}^{\tilde{\theta}}\pi(\chi(y,\alpha(y))) dy-\alpha(\tilde{\theta})\chi(\tilde{\theta},\alpha(\tilde{\theta}))- \\
	\nonumber
	&\tilde{\theta}\pi(\chi(\tilde{\theta},\alpha(\tilde{\theta})))  +p\chi(\tilde{\theta},\alpha(\tilde{\theta}))+ \tilde{\theta}\pi(\chi(\tilde{\theta},\alpha(\tilde{\theta}))) -p\chi(\tilde{\theta},\alpha(\tilde{\theta}))	 \\
	\nonumber
	&\geq \alpha({\check{\theta}})\chi(\tilde{\theta},\alpha({\check{\theta}}))+\int_{\underline{\theta}}^{{\check{\theta}}}\pi(\chi(y,\alpha(y))) dy-\alpha({\check{\theta}})\chi({\check{\theta}},\alpha({\check{\theta}}))- \\
	\label{IC_3}
	& {\check{\theta}}\pi(\chi({\check{\theta}},\alpha({\check{\theta}}))) +p\chi({\check{\theta}},\alpha({\check{\theta}}))+ {\tilde{\theta}}\pi(\chi(\tilde{\theta},\alpha({\check{\theta}}))) -p\chi(\tilde{\theta},\alpha({\check{\theta}}))
	\end{align}
	where ${\check{\theta}}\in  [\underline{\theta}, \bar{\theta}]$ is the arbitrary announced type by the agent.
	By adding and subtracting $\beta(\check{\theta})$ to the left side of  \eqref{IC_3}  and simplifying it, we have
			\begin{align}
	\label{IC_if1}
	&\int_{{\check{\theta}}}^{\tilde{\theta}} \pi(\chi(y,\alpha(y))) dy	\geq \\
 \nonumber
&U(\tilde{\theta},\chi(\tilde{\theta},\alpha({\check{\theta}})),\alpha({\check{\theta}}),\beta({\check{\theta}}))- U(\check{\theta},\chi(\check{\theta},\alpha({\check{\theta}})),\alpha({\check{\theta}}),\beta({\check{\theta}})).
	\end{align}
	Thus, by considering \eqref{Eq:IR_pr2} we have
		\begin{align}
		\label{IC_if2}
&\int_{{\check{\theta}}}^{\tilde{\theta}}[\frac{d U(\theta,\chi(\theta,\alpha(y)),\alpha(y),\beta(y))}{d \theta} \Big|_{\theta=y}]dy \geq	\\
 \nonumber
&U(\tilde{\theta},\chi(\tilde{\theta},\alpha({\check{\theta}})),\alpha({\check{\theta}}),\beta({\check{\theta}}))- U(\check{\theta},\chi(\check{\theta},\alpha({\check{\theta}})),\alpha({\check{\theta}}),\beta({\check{\theta}})).
	\end{align}
	Next, we show that \eqref{IC_if2} which is equivalent to IC constraint holds true. By derivation from \eqref{Eq:IR_pr2} with respect to $\hat{\theta}$, we have
			\begin{align}
	\label{IC_if3}
\frac{d^2 U(\theta,\chi(\theta,\alpha(\hat{\theta})),\alpha(\hat{\theta}),\beta(\hat{\theta}))}{d \hat{\theta}d \theta}=\frac{d \pi}{d \hat{\theta}}=\big(\frac{\partial \pi}{\partial x}\big)_\chi \times\frac{\partial \chi }{\partial \alpha} \times \frac{\partial \alpha}{\partial \hat{\theta}}\ge 0.
	\end{align}
Since $\frac{\partial \pi}{\partial x}$,$\frac{\partial \chi }{\partial \alpha}$ and $\frac{\partial \alpha}{\partial \theta} $ are positive as the results of {Assumption \ref{as:conc}}, Remark \ref{rem1}, and constraint \eqref{c1_2}, respectively, \eqref{IC_3} holds true. 
Hence, 	if $\tilde{\theta} > {\check{\theta}}$, we have 
	\begin{align}
\label{IC_if4}
&\int_{{\check{\theta}}}^{\tilde{\theta}}[\frac{d U(\theta,\chi(\theta,\alpha(y)),\alpha(y),\beta(y))}{d \theta} \Big|_{\theta=y}]dy	\geq \\
\nonumber
&\int_{{\check{\theta}}}^{\tilde{\theta}}\frac{d U(\theta,\chi(\theta,\alpha({\check{\theta}})),\alpha({\check{\theta}}),\beta({\check{\theta}}))}{d \theta} d\theta=
\\
\nonumber &U(\tilde{\theta},\chi(\tilde{\theta},\alpha({\check{\theta}})),\alpha({\check{\theta}}),\beta({\check{\theta}}))- U(\check{\theta},\chi(\check{\theta},\alpha({\check{\theta}})),\alpha({\check{\theta}}),\beta({\check{\theta}}))
\end{align}
and thus \eqref{IC_if2} is verified and hence, the IC constraint holds true.
In a similar way, we can show that IC constraint holds true for  $\tilde{\theta}< {\check{\theta}}$.

In the second part of the proof, we show that given an optimal solution to optimization \eqref{opt_1_1}, we can find a solution to optimization \eqref{opt_2_2}
with the same value of the objective. 
As the first step, we prove that the IC constraint in 
optimization \eqref{opt_1_1} implies the monotonicity of $\alpha(\theta)$. 
Let's consider an agent of type $\tilde{\theta}$ with an announced type $\check{\theta}=\Tilde{\theta}-\epsilon$, where $\epsilon> 0$ and then $\epsilon\to 0$, hence
the IC constraint gives
\begin{align}
\label{IC_5}
&(\alpha({\tilde{\theta}})-p) \chi(\tilde{\theta},\alpha(\tilde{\theta}))+\beta({\tilde{\theta}})+ \tilde{\theta} \pi(\chi(\tilde{\theta},\alpha(\tilde{\theta})))\ge \\
\nonumber
&(\alpha({\check{\theta}})-p) \chi(\tilde{\theta},\alpha(\check{\theta}))+\beta({\check{\theta}})+ \tilde{\theta} \pi(\chi(\tilde{\theta},\alpha(\check{\theta}))
\end{align}
and if we consider an agent of type $\check{\theta}$ with the announced type equal to $\tilde \theta$, then 
the IC constraint reads as
\begin{align}
\label{IC_6}
& (\alpha({\check{\theta}})-p) \chi(\check{\theta},\alpha(\check{\theta}))+\beta({\check{\theta}})+ \check{\theta} \pi(\chi(\check{\theta},\alpha(\check{\theta})))\ge \\
\nonumber
&(\alpha({\tilde{\theta}})-p) \chi(\check{\theta},\alpha(\tilde{\theta}))+\beta({\tilde{\theta}})+ \check{\theta} \pi(\chi(\check{\theta},\alpha(\tilde{\theta}))).
\end{align}	
{By summation of \eqref{IC_5} and \eqref{IC_6} and rearranging the terms, we get
\begin{align}
\nonumber
&(\alpha({\tilde{\theta}})-p) \Big(\chi(\tilde{\theta},\alpha(\tilde{\theta}))-\chi(\check{\theta},\alpha(\tilde{\theta}))\Big)+ \tilde{\theta} \Big(\pi(\chi(\tilde{\theta},\alpha(\tilde{\theta})))-\\
\nonumber
&\pi(\chi(\tilde{\theta},\alpha(\check{\theta})))\Big)\ge \\
\nonumber
&(\alpha({\check{\theta}})-p) \Big(\chi(\tilde{\theta},\alpha(\check{\theta}))-\chi(\check{\theta},\alpha(\check{\theta}))\Big)+ \check{\theta} \Big(\pi(\chi(\check{\theta},\alpha(\tilde{\theta})))-\\
\label{IC_add}
&\pi(\chi(\check{\theta},\alpha(\check{\theta})))\Big).
\end{align}	
}
{By dividing \eqref{IC_add} by $\epsilon$, we have
\begin{align}
\label{IC_add2}
&(\alpha({\tilde{\theta}})-p) \big(\frac{\partial \chi}{\partial {\theta}}\big)_{\tilde{\theta}}+\tilde{\theta} \big(\frac{\partial \pi}{\partial {\hat{\theta}}}\big)_{\tilde{\theta}}\ge 
(\alpha({\check{\theta}})-p) \big(\frac{\partial \chi}{\partial {\theta}}\big)_{\tilde{\theta}}+\check{\theta} \big(\frac{\partial \pi}{\partial {\hat{\theta}}}\big)_{\tilde{\theta}}.
\end{align}	
Equation \eqref{IC_add2} can be rewritten as
\begin{align}
\label{IC_add3}
&(\alpha({\tilde{\theta}})-\alpha({\check{\theta}})) \big(\frac{\partial \chi}{\partial {\theta}}\big)_{\tilde{\theta}}+(\tilde{\theta}-\check{\theta})  \big(\frac{\partial \pi}{\partial x}\big)_\chi \frac{\partial \chi}{\partial {\alpha}}\big(\frac{\partial \alpha}{\partial {\hat{\theta}}}\big)_{\tilde{\theta}}\ge 0.
\end{align}	
Dividing \eqref{IC_add3} again by $\epsilon$, gives
\begin{align}
\label{IC_add4}
&\big(\frac{\partial \alpha}{\partial {\hat{\theta}}}\big)_{\tilde{\theta}}\Bigg(
\big(\frac{\partial \chi}{\partial {\theta}}\big)_{\tilde{\theta}}+ \big(\frac{\partial \pi}{\partial x}\big)_\chi \frac{\partial \chi}{\partial {\alpha}}\Bigg)\ge 0.
\end{align}	
Since $\frac{\partial \pi}{\partial x},\frac{\partial \chi}{\partial \alpha}$ and $\frac{\partial \chi}{\partial \theta}$ are positive as the results of {Assumption \ref{as:conc}}, Remark \ref{rem1}, and the fact that ${\tilde{ \theta}}$ is any arbitrary point in $[\underline{\theta}, \bar{\theta}]$, we can conclude that $\dot{\alpha} ({\hat{\theta}})=\frac{\partial \alpha}{\partial {\hat{\theta}}} \ge 0$.}

To derive constraint \eqref{c2_2}, we rearrange \eqref{IC_5} and \eqref{IC_6} as follows
\begin{align}
\nonumber
&(\alpha(\check{\theta})-p)\chi(\tilde{\theta},\alpha(\check{\theta}))+\tilde{\theta}\pi(\chi(\tilde{\theta},\alpha(\check{\theta})))-(\alpha(\tilde{\theta})-p)\chi(\tilde{\theta},\alpha(\tilde{\theta}))\\
\nonumber
&-\tilde{\theta}\pi(\chi(\tilde{\theta},\alpha(\tilde{\theta}))) \leq \beta(\tilde{\theta})-\beta(\check{\theta}) \leq (\alpha(\check{\theta})-p)\chi(\check{\theta},\alpha(\check{\theta}))\\
\label{onlyif}
&+\check{\theta}\pi(\chi(\check{\theta},\alpha(\check{\theta})))-(\alpha(\tilde{\theta})-p)\chi(\check{\theta},\alpha(\tilde{\theta}))-\check{\theta}\pi(\chi(\check{\theta},\alpha(\tilde{\theta}))).
\end{align}	
Dividing \eqref{onlyif} by $\epsilon$, 
we have
\begin{align}
\nonumber 
&[\frac{d}{d \hat{\theta}}[(\alpha(\hat{\theta})-p)\chi(\theta,\alpha(\hat{\theta}))+\theta\pi(\chi(\theta,\alpha(\hat{\theta})))]]\Big|_{\substack{
		{\hat{\theta}=\tilde{\theta}}\\
		{{\theta}=\tilde{\theta}}
}} \le \frac{d}{d\tilde{\theta}} \beta(\tilde{\theta})  \\
\label{IC_form_d}
&\le[\frac{d}{d \hat{\theta}}[(\alpha(\hat{\theta})-p)\chi(\theta,\alpha(\hat{\theta}))+\theta\pi(\chi(\theta,\alpha(\hat{\theta})))]]\Big|_{\substack{
		{\hat{\theta}=\tilde{\theta}}\\
		{{\theta}=\tilde{\theta}}
}}
\end{align}	
which implies that
\begin{align}
\label{IC_form_d2}
\frac{d}{d\tilde{\theta}} \beta(\tilde{\theta}) =
[\frac{d}{d \hat{\theta}}[(\alpha(\hat{\theta})-p)\chi(\theta,\alpha(\hat{\theta}))+\theta\pi(\chi(\theta,\alpha(\hat{\theta})))]]\Big|_{\substack{
		{\hat{\theta}=\tilde{\theta}}\\
		{{\theta}=\tilde{\theta}}}}.
\end{align}	 
Integrating \eqref{IC_form_d2} with respect to $\tilde{\theta}$ from $\underline{\theta}$ to $\tilde{\theta}$, we have
\begin{align}
\label{Eq:beta}
&\beta(\tilde{\theta})-\beta(\underline{\theta}) =\\
\nonumber
&\int_{\underline{\theta}}^{\theta}	 [\frac{d}{d \hat{\theta}}[(\alpha(\hat{\theta})-p)\chi(\theta,\alpha(\hat{\theta}))+\theta\pi(\chi(\theta,\alpha(\hat{\theta})))]]\Big|_{\substack{
		{\hat{\theta}=\tilde{\theta}}\\
		{{\theta}=\tilde{\theta}}
}}d \tilde{\theta}.
\end{align}
Considering $U(\underline{\theta},\underline{\theta}) = 0$ from Proposition \ref{Lem1} and
the integration by parts
, \eqref{Eq:beta} can be rewritten as follows
\begin{align}
\nonumber
\beta(\tilde{\theta})= &\int_{\underline{\theta}}^{\tilde{\theta}}\big[\frac{d}{d {\theta}}[(\alpha(\hat{\theta})-p)\chi(\theta,\alpha(\hat{\theta}))+\theta\pi(\chi(\theta,\alpha(\hat{\theta})))]\big]\Big|_{\hat{\theta}=\theta}d \theta\\
\label{IC_form_dd}
&-[(\alpha(\tilde{\theta})-p)\chi(\tilde{\theta},\alpha(\tilde{\theta}))+\tilde{\theta}\pi(\chi(\tilde{\theta},\alpha(\tilde{\theta})))].
\end{align}
Considering the definition of $K$, equation \eqref{IC_form_dd} can be rewritten as follows which completes the proof.
\begin{align}
\label{def_beta}
&\beta(\tilde{\theta})=   \int_{\underline{\theta}}^{\tilde{\theta}}K_{\theta}(\chi(\theta,\alpha(\hat{\theta})),\theta,\alpha(\hat{\theta}))\Big|_{\hat{\theta}=\theta}d\theta-\\
\nonumber
&K(\chi(\tilde{\theta},\alpha(\tilde{\theta})),\tilde{\theta},\alpha(\tilde{\theta})) .
\end{align}
\end{proof}
%
%
%
Substituting $\beta(\tilde{\theta})$ from \eqref{def_beta} into the cost function of \eqref{opt_2_2} results in an optimization with double integrals. In order to simplify and solve this optimization problem, we make the following assumption.
\begin{definition}
	\label{def:hazard}
	As is customary in the literature \cite{nisanalgorithmic}, $h(\theta) \equiv f(\theta) / (1-F(\theta))$ denotes as the hazard rate  of type $\theta$.
\end{definition}
\begin{assumption}
	\label{assump1}
	For each ${\theta}  \in  [\underline{\theta}, \bar{\theta}]$ , $h({\theta})$ is increasing 
	\cite{nisanalgorithmic}. 
\end{assumption}
\begin{proposition}
	\label{TH2}
	Optimization \eqref{opt_2_2} is equivalent to
	\begin{subequations}
	\label{opt_3_3}
\begin{align}
			\label{opt_3_1}
			& \max_{ \alpha(\hat{\theta})} \int_{\underline{\theta}}^{\bar{\theta}} \bigg[g(\chi(\hat{\theta},\alpha(\hat{\theta}))) -[K_{\theta}(\chi(\theta,\alpha(\hat{\theta})),\theta,\alpha(\hat{\theta}))]\Bigg|_{\substack{
									{{\theta}=\hat{\theta}}
			}}  \frac{1 }{h(\hat{\theta})} \\
			\nonumber
			&+  \hat{\theta}\pi(\chi(\hat{\theta},\alpha(\hat{\theta}))) -p\chi(\hat{\theta},\alpha(\hat{\theta})) \bigg]f(\hat{\theta})d\hat{\theta}\\
   			\label{const_1}
			&s.t. \; 
			 \dot{\alpha }({\hat{\theta}}) \ge 0\\
			\label{const_2}
			&\chi(\theta,\alpha(\hat{\theta}))=\Gamma(\frac{p-\alpha(\hat{\theta})}{\theta}).
		\end{align}
			\end{subequations}
\end{proposition}
\begin{proof}
	Replacing $\beta(\hat{\theta})$ from \eqref{c2_2} in the utility function of task publisher, we have
	\begin{align}
		\label{opt_4}
	&	\mathbb{E}_{\theta}[V]= \int_{\underline{\theta}}^{\bar{\theta}} \Bigg[g(\chi(\hat{\theta},\alpha(\hat{\theta})))-\alpha(\hat{\theta})\chi(\hat{\theta},\alpha(\hat{\theta}))\\
		\nonumber
		&	-\int_{\underline{\theta}}^{\hat{\theta}}[K_{\theta}(\chi(\theta,\alpha(\hat{\theta})),\theta,\alpha(\hat{\theta}))]\Big|_	{\hat{\theta}={\theta}}d{\theta}\\
  \nonumber
&	+K(\chi(\theta,\alpha(\hat{\theta})),\theta,\alpha(\hat{\theta}))\Big|_{\substack{
			{\theta={\hat{\theta}}}
	}})\Bigg]f(\hat{\theta})d\hat{\theta}.
	\end{align}
	The integration by parts of the term
	$	\int_{\underline{\theta}}^{\bar{\theta}} \Big[\int_{\underline{\theta}}^{\tilde{\theta}}K_{\theta}(X(\theta,\alpha(\hat{\theta})),\theta,\alpha(\hat{\theta}))\Big|_	{\hat{\theta}={\theta}}d{\theta}\Big]f(\tilde{\theta})d\tilde{\theta}$ gives
	\begin{align}
		\label{hazard_1}
		&\int_{\underline{\theta}}^{\bar{\theta}} \Big[\int_{\underline{\theta}}^{\tilde{\theta}}K_{\theta}(\chi(\theta,\alpha(\hat{\theta})),\theta,\alpha(\hat{\theta}))\Big|_	{\hat{\theta}={\theta}}d{\theta}\Big]f(\tilde{\theta})d\tilde{\theta} =\\
		\nonumber &\int_{\underline{\theta}}^{\bar{\theta}} [K_{\theta}(\chi(\theta,\alpha(\hat{\theta})),\theta,\alpha(\hat{\theta}))]\Big|_{\substack{
							{{\theta}=\hat{\theta}}
		}}\frac{1-F(\tilde{\theta}) }{f(\tilde{\theta})} f(\tilde{\theta})d\tilde{\theta}.
	\end{align}
	Considering the definition of $h(t)$  in Definition \ref{def:hazard}, \eqref{hazard_1} can be rewritten as cost function \eqref{opt_3_1} which completes the proof.

\end{proof}
In what follows, we investigate the solution of optimization \eqref{opt_3_3} by rewriting it as an optimal control problem. We also add the following assumption on the derivative of function  $\alpha$.
\begin{assumption}
	\label{upper_bound}
	There exists $\bar u>0$ such that $\forall \hat{\theta}  \in  [\underline{\theta}, \bar{\theta}]$, $\dot{\alpha}(\hat{\theta})\le \bar u$. 
\end{assumption}
Optimization \eqref{opt_3_3} can be rewritten in the form of an optimal control problem with state variable $\alpha(\hat{\theta})$ and control variable $u(\hat{\theta})$ as follows
	\begin{subequations}
	\label{opt_4_4}
		\begin{align}
\label{opt_4_1}
& \max_{ \alpha(\hat{\theta})} \int_{\underline{\theta}}^{\bar{\theta}}V_{sp} d\hat{\theta}\\
\label{const_11}
s.t.\; \; \;  &\dot{\alpha}  ({\hat{\theta}}) = u(\hat{\theta})\\
\label{const_22}
& u(\hat{\theta}) \in \mathbb{U} := [0,\bar{u}]\\
\label{const_33}
&\chi(\theta,\alpha(\hat{\theta}))=\Gamma(\frac{p-\alpha(\hat{\theta})}{\theta})
\end{align}
	\end{subequations}
where $V_{sp} =\big[g(\chi(\hat{\theta},\alpha(\hat{\theta}))) -[K_{\theta}(\chi(\theta,\alpha(\hat{\theta})),\theta,\alpha(\hat{\theta}))]\Big|_{\substack{
		{{\theta}=\hat{\theta}}
}}  \frac{1 }{h(\hat{\theta})} +  \hat{\theta}\pi(\chi(\hat{\theta},\alpha(\hat{\theta}))) -p\chi(\hat{\theta},\alpha(\hat{\theta})) \big]f(\hat{\theta})$. We introduce the Hamiltonian function as follows \cite{bertsekas2012dynamic}
	\begin{align}
\label{hamiltonian}
H(\alpha,u,\lambda)=V_{sp}(\alpha)+\lambda u
\end{align}
where $\lambda$ is a Lagrange multiplier. Next, we present a proposition that provides both necessary and sufficient optimality conditions for \eqref{opt_4_4} and then, a numerical algorithm proposed in Algorithm \ref{alg1} to find the solution.

\begin{proposition}
The control and state functions $u(\hat{\theta})$,  $\alpha(\hat{\theta})$ are the solution of the optimization \eqref{opt_4_4} if and only if the following conditions which are known as Minimum Principle are met.
		\begin{align}
	\label{cond1}
	& \dot{ \alpha} ({\hat{\theta}}) =\frac{\partial }{\partial \lambda}H(\alpha,u,\lambda)\\
	\label{cond2}
	& \dot{\lambda} ({\hat{\theta}}) =\frac{\partial }{\partial \alpha}H(\alpha,u,\lambda)\\
	\label{cond3}
	& u(\hat{\theta}) = \arg \min_{u\in \mathbb{U}} H(\alpha,u,\lambda)\\
	& \lambda(\bar{\theta})=0. 
	\end{align}
\end{proposition}
\begin{proof}
	As shown in \cite[chapter~3 ]{bertsekas2012dynamic}, since $\dot{\alpha}$ is the linear function of $\alpha$ and $u$ and also $V_{sp}$ is a concave function and $\mathbb{U}$ is a convex set,  the conditions of the minimum principle are both necessary and sufficient for optimality. 
\end{proof}
We consider Gradient Projection Algorithm explained in Algorithm \ref{alg1} to solve our optimal control problem.
In \cite[Proposition~2.4 and Lemma~2.5 ]{preininger2018convergence}, it is shown that if $\sum_{i=0}^{\infty}\gamma^i = \infty$,   and $\lim\limits_{i\to\infty}\gamma^i = 0$, where $\gamma^i$ is the learning rate, then the sequence ${u^i}$ in Algorithm. \ref{alg1} converges to optimal control function.
\begin{algorithm}
	\caption{The Gradient Projection Algorithm for solving optimal control problem that the task publisher is faced to design optimal mechanism}
	\begin{algorithmic}[1]
			\label{alg1}
		\STATE  Select a discrete intial approximation for control variable $u^{0}{(\hat{\theta})}$.
		\STATE Using the control variable $u^i(\hat{\theta})$, solve differential equation \eqref{cond1} with initial condition $\alpha(\underline{\theta})=\alpha_0$ and calculate $\alpha^{i}(\hat{\theta})$. 
		\STATE Using the control variable $u^{i}(\hat{\theta})$ and $\alpha^{i}(\hat{\theta})$ from previous steps , solve differential equation \eqref{cond2} with initial condition $\lambda(\bar{\theta})=0$ and calculate $\lambda^{(i)}(\hat{\theta})$.
		\STATE Compute $u^{i+1}=P_U(u^i-\gamma^i \frac{\partial H}{\partial u^i})$.
		\STATE  If $u^{i+1}=u^{i}$ then stop. Otherwise, replace $i$ with $i+1$ and go to step 2.
	\end{algorithmic}
\end{algorithm}
\section{illustrative example}
\label{Sec:Simulation}
In this section, we evaluate the performance of the proposed mechanism. We assume that the types of agents are uniformly distributed within the interval $[4,6]$. The cost of completing one unit of a task is set at $p=10$. The functions $\pi(x)=\frac{z_2}{1-z_1}x^{1-z_1}$ and $g(x)=\frac{1}{1-q_1}x^{1-q_1}$ are used, with the values $q_1=0.5$, $z_1=0.5$, and $q_2=3$ assigned. 
The initial control variable is set to $u^0(\hat{\theta})=0.5$ for all $\hat{\theta}\in [\underline{\theta}, \bar{\theta}]$, and the learning rate is $\gamma^i=0.01$. Using these functions and values, we can calculate $\chi$ using \eqref{d11} as $\chi(\theta,\alpha(\hat{\theta}))=(\frac{p-\alpha(\hat{\theta})}{\theta})^{-\frac{1}{z_1}}$. The convergence of Algorithm \ref{alg1} is depicted in Figure \ref{Fig:alg}, which shows that $\alpha(\hat{\theta})$ converges after approximately 300 iterations.

To demonstrate the validity of the IR and IC constraints in the proposed scheme, the utilities of five specific agent types ($\theta=4,4.5,5,5.5,6$) when they announce different types are shown in Figure\ref{Fig:type}. The black stars on the curve represent the points at which each type of agent obtains maximum utility. Figure \ref{Fig:type} illustrates that when an agent truthfully announces its type, it receives a positive and maximal utility. However, when the agent deceives the task publisher by announcing a false type, it incurs a loss. These results confirm that 
the proposed mechanism satisfies the IR and IC constraints.
  \begin{figure}[h!]
	\centering
 \vspace*{-3mm}
\includegraphics[width=0.95\linewidth]{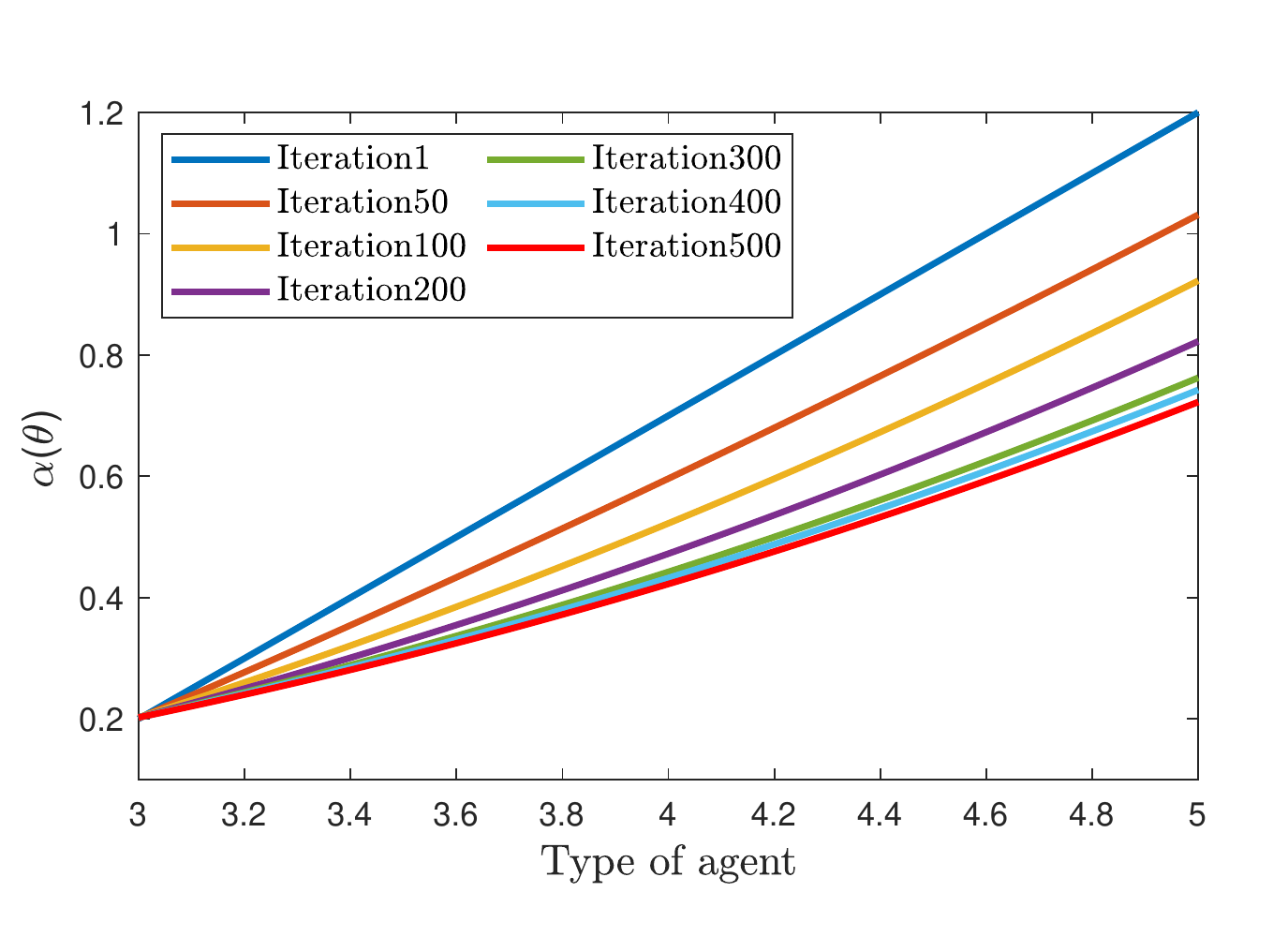}
	\caption{ $\alpha(\hat{\theta})$ allocated to different types of agents in different iterations of Algorithm \ref{alg1}.}
	\label{Fig:alg}
\end{figure}
  \begin{figure}[h!]
	\centering
 \vspace*{-3mm}
\includegraphics[width=0.95\linewidth]{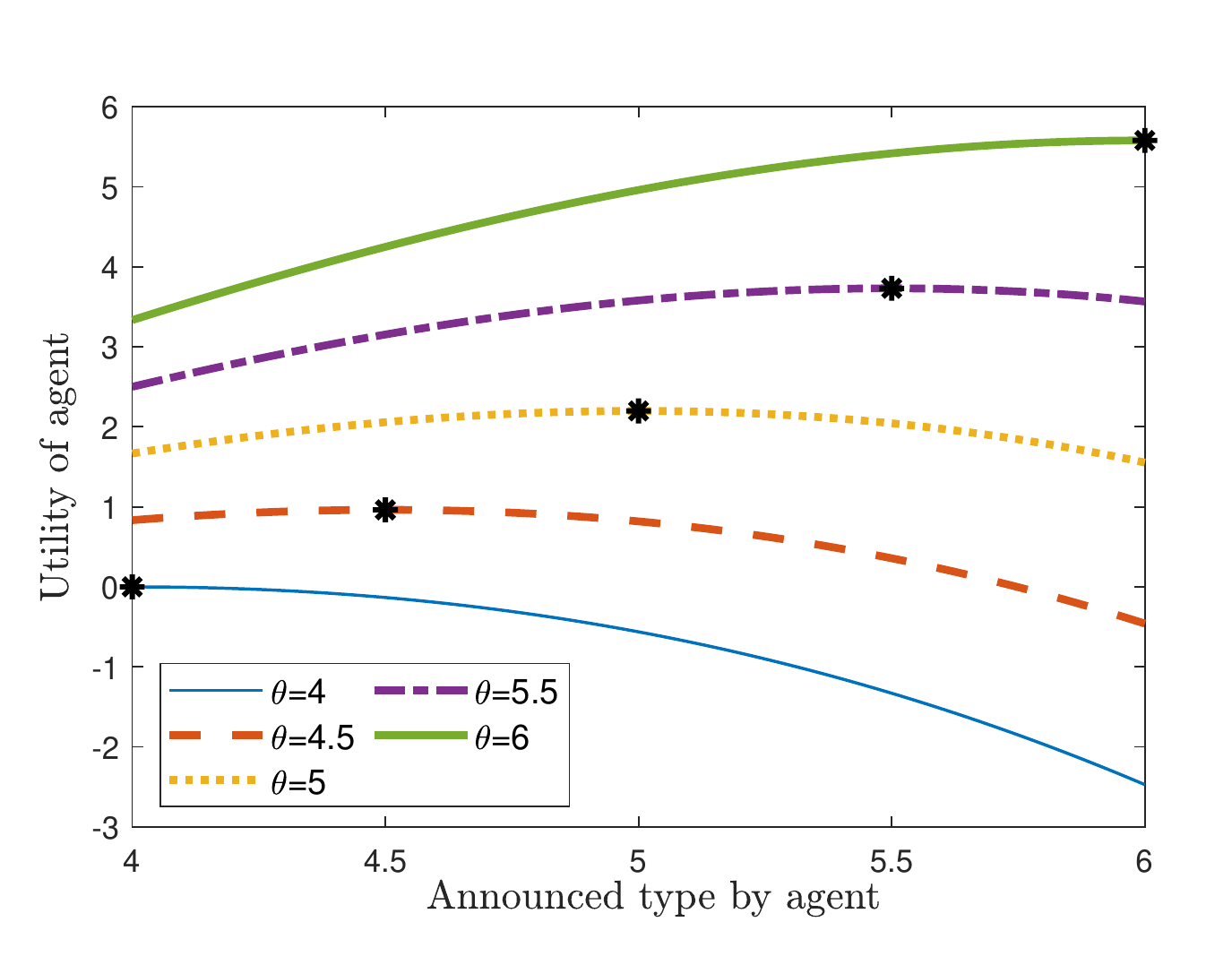}
	\caption{ Utilities of different agents 
		when announcing different types.}
	\label{Fig:type}
	\vspace*{-3mm}
\end{figure}

\section{Conclusion} 
This paper presented a new profit maximization mechanism for task allocation under incomplete information about the utilities of autonomous agents. The proposed mechanism allows the task publisher to maximize its utility and simultaneously, ensures both truthful reporting of the agents' private information and allows autonomous agents to decide on their own participation levels.
{We formulated the optimal truthful mechanism as a nonconvex functional optimization problem. By establishing a relation between the decision functions of the task publisher, we found an equivalent representation of the incentive constraint, which transformed the nonconvex optimization into a tractable convex optimal control problem. }

\label{Sec:Conclusion}

 \bibliographystyle{IEEEtran}
%
\bibliography{generic-color-brief.bib}

\end{document}